\documentclass[a4paper,USenglish]{lipics}
\usepackage{microtype}
\usepackage{tikz}

\usepackage{amsmath}
\usepackage{amsfonts}
\usepackage{amssymb}
\usepackage{amsthm}

\usepackage{graphicx}
\usepackage{float}
\usepackage{xspace}

\newcommand{\df}{\emph}
\newcommand{\subrelated}{\textbf}

\DeclareMathOperator{\poly}{poly}
\DeclareMathOperator{\polylog}{polylog}

\newcommand{\newextmathcommand}[2]{%
    \newcommand{#1}{\ensuremath{#2}\xspace}
}

\newextmathcommand{\A}{\mathcal{A}}
\newextmathcommand{\Net}{\mathcal{N}}
\newextmathcommand{\Graph}{G_\Net}

\newcommand{\sources}{\mathsf{in}}
\newcommand{\sinks}{\mathsf{out}}

\newcommand*{\orparam}[1]{\mathsf{OR_{#1}}}
\newcommand*{\ortwo}{\mathsf{OR_2}}
\newcommand*{\orany}{\mathsf{OR}}

\newcommand*{\rankor}{\mathsf{rk_\lor}}
\newcommand*{\orrank}{\rankor}
\newcommand*{\frrank}{\mathsf{rk_\lor^*}}

\newcommand*{\sumparam}[1]{\mathsf{SUM_{#1}}}
\newcommand*{\sumany}{\mathsf{SUM}}
\newcommand*{\sumtwo}{\mathsf{SUM_2}}

\newcommand*{\indices}[1]{\{1, \ldots, #1\}}
\newextmathcommand{\setn}{\indices{n}}
\newextmathcommand{\setm}{\indices{m}}
\newextmathcommand{\setk}{\indices{k}}

\newcommand*{\indicez}[1]{\{0, \ldots, #1\}}
\newextmathcommand{\kbits}{\indicez{n - 1}}

\newextmathcommand{\Bin}{\{0, 1\}}
\newextmathcommand{\OR}{\mathrm{OR}}
\newextmathcommand{\NOT}{\mathrm{NOT}}
\newextmathcommand{\AND}{\mathrm{AND}}
\newextmathcommand{\SUM}{\mathrm{SUM}}

\newextmathcommand{\Family}{\mathcal F}
\newextmathcommand{\SubFamily}{\mathcal F'}

\newextmathcommand{\OptCover}{\tau}
\newextmathcommand{\OptFrac}{\tau^*}

\newextmathcommand{\LT}{\textsc{Less-Than}}

\newcommand*{\compl}{\overline}

\newcommand{\N}{\mathbb N}

\newcommand*{\eps}{\varepsilon}

\newcommand{\sset}{\subseteq}

\newcommand{\mynote}[1]{}
\newcommand{\nvts}{\negthickspace\negthickspace\negthickspace}

\newcommand{\seqnum}[1]{\href{http://oeis.org/#1}{\underline{#1}}}

\newtheorem{myremark}[theorem]{Remark}

\def\isetn{{\cal I}_n}

\theoremstyle{remark}
\newtheorem*{claim*}{Claim}

\title{Fractional coverings, greedy coverings, and rectifier networks}

\author[1]{Dmitry Chistikov}
\author[2]{Szabolcs Iv\'an}
\author[3]{Anna Lubiw}
\author[3]{\mbox{Jeffrey Shallit}}

\affil[1]{Max Planck Institute for Software Systems (MPI-SWS), Germany, \texttt{dch@mpi-sws.org}}
\affil[2]{University of Szeged, Hungary, \texttt{szabivan@inf.u-szeged.hu}}
\affil[3]{School of Computer Science, University of Waterloo, Canada, \texttt{\{alubiw,shallit\}@cs.uwaterloo.ca}}

\serieslogo{}
\volumeinfo
  {Billy Editor and Bill Editors}
  {2}
  {Conference title on which this volume is based on}
  {1}
  {1}
  {1}
\EventShortName{}
\DOI{10.4230/LIPIcs.xxx.yyy.p}

\begin{document}

\makeatletter
\renewcommand*{\@Copyright}{}
\makeatother

\maketitle

\begin{abstract}
A rectifier network is a directed acyclic graph with distinguished sources and sinks;
it is said to compute a Boolean matrix $M$ that has a~$1$ in the entry~$(i,j)$
iff there is a path from the $j$th source to the $i$th sink.
The smallest number of edges in a rectifier network that computes $M$
is a classic complexity measure on matrices, which has been studied for more than half a century.

We explore two well-known techniques that have hitherto found little to no applications
in this theory.
Both of them build upon a basic fact that depth-$2$ rectifier networks are essentially
weighted coverings of Boolean matrices with rectangles.
We obtain new results
by using \emph{fractional} and \emph{greedy} coverings (defined in the standard way).

First, we show that all \emph{fractional} coverings of the so-called full triangular matrix
have cost at least $n \log n$.
This provides (a fortiori) a new proof of the tight lower bound
on its depth-$2$ complexity (the exact value has been known since 1965, but
previous proofs are based on different arguments).
Second, we show that the \emph{greedy} heuristic is instrumental in tightening the upper
bound on the depth-$2$ complexity of the Kneser-Sierpi\'nski (disjointness) matrix.
The previous upper bound is $O(n^{1.28})$, and we improve it to $O(n^{1.17})$, while
the best known lower bound is $\Omega(n^{1.16})$.
Third, using \emph{fractional} coverings,
we obtain a form of direct product theorem that gives a lower bound on unbounded-depth complexity
of Kronecker (tensor) products of matrices.
In this case, the \emph{greedy} heuristic shows (by an argument due to Lov\'asz)
that our result is only a logarithmic factor away from the ``full''
direct product theorem.
Our second and third results
constitute progress on open problem~7.3 and
resolve, up to a logarithmic factor, open problem~7.5 from a recent book by
Jukna and Sergeev (in~Foundations and Trends in Theoretical Computer Science~(2013)).

\end{abstract}

\section{Introduction}

Introduced in the 1950s,
\df{rectifier networks} are one of the oldest and most basic models
in the theory of computing.
They are directed acyclic graphs with distinguished input and output
nodes; a rectifier network is said to \df{compute} (or \df{express})
the Boolean matrix $M$ that has a~$1$ in the entry~$(i,j)$ iff there is a path
from the $j$th input to the $i$th output.
Equivalently, rectifier networks can be viewed as Boolean circuits that consist
entirely of \OR gates of arbitrary fan-in.
This simple model of computation has attracted a lot of attention~\cite{js13},
because it captures the ``topological'' core of other models:
complexity bounds for rectifier networks
extend in one way or another to Boolean circuits (i.e., circuits with Boolean gates)
and to switching circuits~\cite{nechiporuk2,Mehlhorn79}.

Given a matrix $M$, what is the smallest number of edges in a rectifier network
that computes $M$? Denote this number by $\orany(M)$---this is a complexity
measure on Boolean matrices.
This measure is fairly well understood:
we know, from Nechiporuk~\cite{nechiporuk},
that the maximum of $\orany(M)$ grows as $n^2 / 2 \log n$ as $n \to \infty$
if $M$ is $n \times n$;
we also know that random $n \times n$-matrices have complexity very
close to $n^2 / 2 \log n$.
The ``shape'' of these two facts is reminiscent of
the standard circuit complexity of Boolean functions over \AND, \OR, and \NOT gates---%
but for them, the maximum is $2^n / n$
instead of $n^2 / 2 \log n$.

However, much more is known about the measure $\orany(\cdot)$:
there are explicit sequences of matrices that have complexity $n^{2 - o(1)}$,
close to the maximum
(in contrast, for circuits over AND, OR, and NOT gates, exhibiting a single sequence
of functions that require a superlinear number of gates would be a tremendous
breakthrough).
%
In fact, nowadays a range of methods are available for obtaining upper
and lower bounds on $\orany(M)$ for specific matrices $M$;
we refer the interested reader to the recent book by Jukna and Sergeev~\cite{js13}.

Many natural questions, however, remain open.
Jukna and Sergeev list 19 open problems about $\orany(\cdot)$
and related complexity measures.
Several of them refer to very restricted submodels, such as
rectifier networks of depth~$2$: that is, networks where all paths
contain (at most) $2$~edges.
A depth-$2$ rectifier network expressing a matrix $M$
is essentially a \df{covering} of $M$---a collection of
(rectangular) all-$1$ submatrices of $M$ whose disjunction is~$M$.
In our work, we look into the corresponding complexity measure~$\ortwo(\cdot)$
as well as~$\orany(\cdot)$. We build upon the connection between rectifier networks
and (weighted) set coverings and
explore two well-known ideas that have previously found
few applications in the study of rectifier networks:
they are associated with fractional and greedy coverings respectively.

\emph{Fractional coverings}
are a generalization of usual set coverings.
In the usual set cover problem, each set~$S$ can be either included or not
included in the solution (i.e., in the covering);
in the fractional version each set can be partially included:
a solution assigns to each set~$S$ a real number $x_S \in [0; 1]$,
and for every element~$s$ of the universe the sum $\sum_{s \in S} x_S$
should be equal to or exceed~$1$.
In other words, fractional coverings
arise
from linear relaxation of the integer program that expresses the set cover problem.
\emph{Greedy coverings} are, in contrast, usual coverings;
they are the outcome of applying the standard greedy heuristic
to an instance of the set cover problem:
at each step, the algorithm picks a set~$S$ that
covers the largest number of yet uncovered elements~$s$.
In our work, we use fractional and greedy coverings
to obtain estimates on the values of $\ortwo(M)$
and $\orany(M)$.

\subsubsection*{Our results}
%
First, we demonstrate that
$\ortwo(T_n) = n(\lfloor \log_2 n \rfloor + 2) - 2^{\lfloor \log_2 n \rfloor + 1}$,
where $T_n$ is the so-called full triangular matrix: an upper-triangular
matrix that has $1$s everywhere above the main diagonal and $0$s on the diagonal and below.
In this problem, the upper bound is easy and the challenge is to prove the lower bound.
This was previously done by Krichevskii~\cite{kr-pk65}, and our paper provides a different
proof of independent interest. In fact, we prove a stronger statement:
all \emph{fractional} coverings of $T_n$ have large associated cost
(Theorem~\ref{th:fft}).
To this end, we take the linear program that expresses the fractional set cover problem and
find a good feasible solution to the dual program.
The value of this solution then gives a lower bound on the cost of all feasible
solutions to the primal---that is, on the cost of fractional coverings.
Since integral coverings are just a special case of fractional coverings,
the result follows.

Second, we improve the upper bound on the value of $\ortwo(D_n)$, where $D_n$ is
the disjointness matrix, also known as the Kneser-Sierpi\'nski matrix.
This constitutes progress on open problem~7.3
in Jukna and Sergeev's book~\cite{js13}, where the previously known
bounds are obtained.
The previous upper bound is $O(n^{1.28})$,
and our Theorem~\ref{th:ks} improves it to $O(n^{1.17})$, while
the best known lower bound is $\Omega(n^{1.16})$.
To achieve this improvement, we subdivide the instance of the weighted set cover problem
(in which the optimal value is $\ortwo(D_n)$)
into $\polylog(n)$ natural subproblems
and reduce them, by imposing an additional restriction, to
instances of unweighted set cover problems.
We then solve these instances with the \emph{greedy} heuristic;
the upper bound in the analysis invokes
the so-called greedy covering lemma by Sapozhenko~\cite{sap72},
also known as the Lov\'asz--Stein theorem~\cite{lovaszDM,stein}.
This gives us the desired upper bound on $\ortwo(D_n)$;
in fact, the greedy strategy turns out to be optimal, and
the optimal exponent in $\ortwo(D_n)$ comes from a numerical optimization problem.
As an intermediate result we determine, up to a polylogarithmic factor,
the value of $\ortwo(D^m_k)$
where $D^m_k$ is the adjacency matrix of the Kneser graph
on $2 \binom k m$ vertices.

Finally, we obtain (Theorem~\ref{th:dirprod})
a form of direct product theorem for the $\orany(\cdot)$ measure:
$\orany(K \otimes M) \ge \frrank(K) \cdot \orany(M)$.
Here $K \otimes M$ denotes the Kronecker product of matrices $K$ and $M$,
and $\frrank(K)$ is a fractional analogue of the Boolean rank of $K$.
This resolves, up to a logarithmic factor,
open problem~7.5 in the list of Jukna and Sergeev~\cite{js13},
which asks for the lower bound of $\orrank(K) \cdot \orany(M)$
where $\orrank(K) \ge \frrank(K)$ is the Boolean rank of~$K$.
(In fact, a~related question for unambiguous rectifier networks, or \SUM-circuits,
is originally due to Find~et~al.~\cite{find13};
our technique applies to this model as well,
giving an analogous inequality for the measure~$\sumany(\cdot)$,
see Corollary~\ref{cor:dirprod}.)
Suppose $K$ is an $m\times n$ matrix;
then, by the argument due to Lov\'asz~\cite{lovasz}, the \emph{greedy} 
heuristic shows that
$\frrank(K) \ge \orrank(K) / (1 + \log m n)$,
so our lower bound is indeed at most a logarithmic factor away from
the ``full'' direct product theorem.
To prove our lower bound, we take the linear programming formulation of
the \emph{fractional} set cover problem for the matrix $K$ and use components
of the optimal solution to the dual program to guide our argument.
It is interesting to see how
reasoning about coverings, or, equivalently, about depth-$2$ rectifier networks,
enables us to obtain meaningful lower bounds on the size of rectifier networks that have
unbounded depth.

\section{Discussion and related work}

We use the matrix language in this paper, but all results can be restated
in terms of biclique coverings of bipartite graphs.

The \subrelated{$\ortwo$-complexity of full triangular matrices}, $T_n$,
is tightly related to results on biclique coverings of complete undirected (non-bipartite) graphs
from the early days of the theory of computing.
The $n \log n$ lower bound, in one form or another, was known
to Hansel~\cite{hansel64}, Krichevskii~\cite{kr-pk65},
Katona and Szemer\'edi~\cite{ks67}, and Tarj\'an~\cite{tarjan75}.%
\footnote{Not all of these arguments compute the \emph{exact} value of $\ortwo(T_n)$.}
Apart from purely combinatorial considerations, the interest in this problem is
motivated by its applications
in formula and switching-circuit complexity of the Boolean threshold-$2$ function
(which takes on the value~$1$ if and only if at least two of its inputs are set to~$1$).
For more context, see treatments by Radhakrishnan~\cite{r-entropy} and Lozhkin~\cite{lozhkin05}.
Our lower bound is obtained in a slightly more restrictive setting, because of
explicit asymmetry:
for $\ortwo(T_n)$, one needs to cover entries $(i, j)$ with $i < j$ in the matrix;
in biclique coverings of undirected graphs, it suffices to cover either of $(i, j)$ and $(j, i)$.
Nevertheless, to the best of our knowledge,
ours is the only proof that goes via linear programming (LP) duality
and provides a tight lower bound on the size of \emph{fractional} coverings.
This result is new;
we are not aware of other lower bounds for rectifier networks that
come from feasible solutions to the LP dual
(in approximation algorithms, a related technique is known under the name
of ``dual fitting''~\cite[Section~9.4]{w-book}).

As for the \subrelated{greedy heuristics}, we are not the first to use them
in the context of depth-$2$ rectifier networks.
Andreev~\cite{a95} obtained a tight worst-case upper bound for
a class of matrices potentially containing ``wildcard'' entries ($*$).
This upper bound is
in terms of the number of occurrences of $0$s and $1$s,
provided that these numbers satisfy certain conditions as the matrix size tends to infinity.
Our Theorem~\ref{th:ks}, however, does not follow from Andreev's worst-case bound.
The disjointness matrix, $D_n$, which we apply this technique to,
is a well-studied object in communication complexity~\cite{kushnisan};
it is a discrete version of the Sierpi\'nski triangle.
Boyar and Find~\cite{bf15} and Selezneva~\cite{selezneva13}
proved that $\orany(D_n) = \Theta(n \log n)$ and
$\sumany(D_n) = \frac{1}{2} n \log n$.%
\footnote{Recall that the $\sumany(\cdot)$ measure
corresponds to \df{unambiguous} rectifier networks,
in which every input-output pair is connected
by at most one path; or, equivalently, to arithmetic circuits over nonnegative
integers with addition (\SUM) gates. For any matrix $M$, $\orany(M) \le \sumany(M)$
and $\ortwo(M) \le \sumtwo(M)$.}
In depth~2, the previous bounds are due to Jukna and Sergeev~\cite{js13};
it is unknown if greedy heuristics are also of use for \SUM-circuits,
as our upper bound for $D_n$ does not extend to this model
(our coverings are not partitions).

\subrelated{Direct sum and direct product theorems} in the theory of computing
are statements of the following form:
when faced with several instances of the same problem on different independent
inputs, there is no better strategy than solving each instance independently.%
\footnote{In some contexts, the terms ``direct sum theorem'' and
    ``direct product theorem'' have slightly different meanings~\cite{sherstov},
    but in the current context we do not distinguish between them.}
For rectifier networks, these questions are associated with
the complexity of Kronecker (tensor) products of matrices.
Indeed, denote the $k \times k$-identity matrix by $I_k$, then
$I_k \otimes M$ is the block-diagonal matrix with $k$ copies of $M$
on the diagonal.
It is not difficult to show that
$\orany(I_k \otimes M) \ge k \cdot \orany(M)$,
and a natural generalization asks whether
$\orany(K \otimes M) \ge \orrank(K) \cdot \orany(M)$
for any matrix $K$---%
see Find et~al.~\cite{find13} and
Jukna and Sergeev~\cite[Sections~2.4, 3.6, and open problem~7.5]{js13}.
To date, this inequality is only known to hold in special cases.
For example,
Find et al.~\cite{find13} can show this lower bound when
the matrix $K$ has a fooling set of size $\orrank(K)$;
however, the size of the largest fooling set does not approximate
the Boolean rank, as observed, e.g., by Gruber and Holzer~\cite{GruberH06}
(they use the graph-theoretic language, with bipartite
 dimension instead of $\orrank$).
As another example, denote by $|M|$ the number of $1$s in the matrix $M$ and assume
that $M$ has no all-$1$ submatrices of size $(k+1)\times(l+1)$.
Then the inequality $\orany(M)\geq |M| / k l$ is a well-known lower bound
due to Nechiporuk~\cite{nechiporuk2}, subsequently rediscovered by
Mehlhorn~\cite{Mehlhorn79}, Pippenger~\cite{Pippenger80}, and Wegener~\cite{Wegener80};
Jukna and Sergeev~\cite[Theorem~3.20]{js13} extend it to
$\orany(K\otimes M)\geq\orrank(K) \cdot |M| / k l$ for any square matrix~$K$.
To the best of our knowledge, the current literature has
no stronger lower bounds
on the $\orany$-complexity of Kronecker products;
our Theorem~\ref{th:dirprod} comes logarithmically close to the desired bound.
For $\sumany$-complexity, the state of the art and our contribution
are analogous to the $\orany$-case.
The related notion of a fractional biclique cover
has previously appeared, e.g., in the papers of Watts~\cite{Watts06}
and Jukna and Kulikov~\cite{JuknaK09}.

Also related to our work is the study of the size of smallest biclique
coverings, under the name of the bipartite dimension of a graph
(as opposed to the cost of such coverings and the $\ortwo$-complexity;
 see Section~\ref{s:pre}).
This quantity corresponds to the Boolean rank of a matrix and is known
to be PSPACE-hard to compute~\cite{GruberH06} and NP-hard to approximate to within
a factor of $n^{1 - \eps}$~\cite{ChalermsookHHK14}.
Finally, we note that results on $\ortwo$-complexity have corollaries
for \subrelated{descriptional complexity} of regular languages.
Indeed, take a language where all words have length two,
$L\subseteq\Sigma \cdot \Delta$,
with $\Sigma=\{a_1, \ldots, a_m\}$ and $\Delta=\{a_1, \ldots, a_n\}$.
Let $M^L$ be its characteristic $m\times n$ matrix:
$M_{i,j}^L=1$ iff $a_i \cdot a_j\in L$.
Then $\ortwo(M^L)$ coincides with the alphabetic length of the shortest
regular expression for $L$;
for example, it follows from Corollary~\ref{cor:ft} that
the optimal regular expression for the language $L_n = \{ a_i a_j \mid 1 \leq i < j \leq n \}$
has $n(\lfloor \log_2 n \rfloor + 2) - 2^{\lfloor \log_2 n \rfloor + 1}$ occurrences of
letters ($\Sigma = \Delta = \{ a_1, \ldots, a_n \}$).
The values of $\orany(M^L)$ and $\ortwo(M^L)$ are also related
to the size of the smallest nondeterministic
finite automata accepting $L$; see~\cite{ivanlelkesnagygyorgyturan14}
and Appendix for details.

\section{Rectifier networks and coverings}
\label{s:pre}


\subsubsection*{Rectifier networks}
Define a \df{rectifier network} with $m$ inputs and $n$ outputs
as a 4-tuple $\Net = (V, E, \sources, \sinks)$, where
$V$ is a set of vertices,
$E \sset V^2$ a set of edges such that the directed graph $\Graph = (V, E)$ is acyclic, and
$\sources \colon \setn \to V$ and
$\sinks \colon \setm \to V$ are injective functions
whose images contain only sources (and, respectively, only sinks) of \Graph.
The network \Net is said to have \df{size}~$|E|$.

A rectifier network \Net \df{expresses} a Boolean $m \times n$ matrix $M = M(\Net)$
such that $M_{i j} = 1$ if \Graph contains a directed path from $\sources(j)$
to $\sinks(i)$ and $M_{i j} = 0$ otherwise.
A rectifier network \Net is said to have \df{depth}~$d$ if
all maximal paths in \Graph have exactly $d$ edges.
Given a Boolean matrix $A \in \Bin^{m \times n}$,
let $\ortwo(A)$ denote the smallest size of a depth-$2$ rectifier network
that expresses $A$
and let $\orany(A)$ denote the smallest size of any rectifier network
that expresses $A$.

This notation is justified by the following observation.
A rectifier network \Net may be viewed as a circuit:
its Boolean inputs are located at the vertices $\sources(\setn)$,
and gates at all other vertices
compute the disjunction (Boolean \OR) of their inputs.
From this point of view, the circuit computes a linear operator
over the monoid $(\Bin, \OR)$, and the matrix of this linear operator
is exactly the Boolean matrix expressed by the rectifier network \Net.

\begin{figure}
\centering
\subfloat[Rectifier network of depth 3]{
\label{f:example:network}
\footnotesize
\begin{tikzpicture}%
[
    thick,
    scale=0.6,
    >=stealth,
    dot/.style={
        circle,
        scale=0.35,
        outer sep=2.0
    }
]
\foreach \i in {1,...,8}{
\node[dot] (a\i) [label=below:\i] at (\i,0) [fill] {};
\node[dot] (b\i) [label=above:\i]at (\i,3) [fill] {};
}
\node[dot] (n1) at (2.5,1) [fill] {};
\node[dot] (n2) at (2.5,2) [fill] {};
\node[dot] (n3) at (6.5,1) [fill] {};
\node[dot] (n4) at (6.5,2) [fill] {};
\foreach \i in {1,...,4}{
\path[->] (b\i) edge (n2);
\path[->] (n1) edge (a\i);
}
\foreach \i in {5,...,8}{
\path[->] (b\i) edge (n4);
\path[->] (n3) edge (a\i);
}
\path[->] (n2) edge (n1);
\path[->] (n4) edge (n1);
\path[->] (n4) edge (n3);

\end{tikzpicture}
}
\subfloat[Matrix $B$]{
\raisebox{\depth}{
\makebox[10em][c]{
\label{f:example:matrix}
$B = \left(\begin{smallmatrix}
     1 & 1 & 1 & 1 & 1 & 1 & 1 & 1 \\
     1 & 1 & 1 & 1 & 1 & 1 & 1 & 1 \\
     1 & 1 & 1 & 1 & 1 & 1 & 1 & 1 \\
     1 & 1 & 1 & 1 & 1 & 1 & 1 & 1 \\
     0 & 0 & 0 & 0 & 1 & 1 & 1 & 1 \\
     0 & 0 & 0 & 0 & 1 & 1 & 1 & 1 \\
     0 & 0 & 0 & 0 & 1 & 1 & 1 & 1 \\
     0 & 0 & 0 & 0 & 1 & 1 & 1 & 1 \\
     \end{smallmatrix}\right)$
}
}
}
\subfloat[Rectifier network of depth 2]{
\label{f:example:depth-2}
\footnotesize
\begin{tikzpicture}%
[
    thick,
    scale=0.6,
    >=stealth,
    dot/.style={
        circle,
        scale=0.35,
        outer sep=2.0
    }
]
\foreach \i in {1,...,8}{
\node[dot] (a\i) [label=below:\i] at (\i,0) [fill] {};
\node[dot] (b\i) [label=above:\i]at (\i,3) [fill] {};
}
\node[dot] (n1) at (2.5,1.5) [fill] {};
\node[dot] (n2) at (6.5,1.5) [fill] {};
\foreach \i in {1,...,4}{
\path[->] (b\i) edge (n1);
\path[->] (n1) edge (a\i);
\path[->] (n2) edge (a\i);
}
\foreach \i in {5,...,8}{
\path[->] (b\i) edge (n2);
\path[->] (n2) edge (a\i);
}

\end{tikzpicture}
}
\caption{Illustrations for Example~\ref{our-example}}
\end{figure}
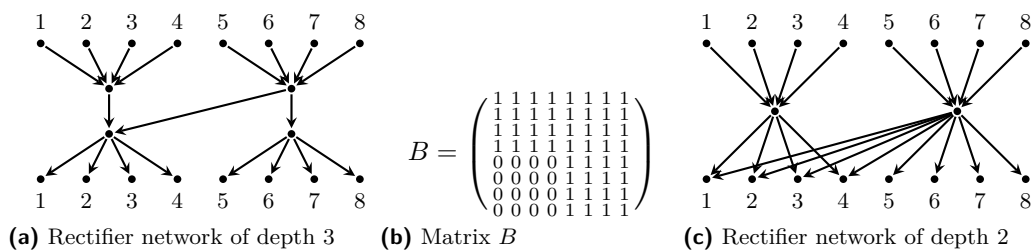

\begin{example}
\label{our-example}
A depth-3 rectifier network is shown in Figure~\ref{f:example:network}.
It expresses the matrix $B$ in Figure~\ref{f:example:matrix}, showing that
$\orparam{3}(B) \le 19$.
In fact, this network is optimal and $\orparam{3}(B) = 19$; see Appendix for details.
At the same time, $\ortwo(B) = 20$: the upper bound is achieved by the network
in Figure~\ref{f:example:depth-2}, and the lower bound is due to
Jukna and Sergeev~\cite[Theorem~3.18]{js13}.
\end{example}

\subsubsection*{Coverings of Boolean matrices}
Let us describe an alternative way of defining the function~$\ortwo(\cdot)$.
Given a Boolean matrix~$A$, a \df{rectangle} (or a $1$-rectangle)
is a pair $(R, C)$, where $R \sset \setm$ and $C \sset \setn$,
such that for all $(i, j) \in R \times C$ we have $A_{i j} = 1$.
A rectangle $(R, C)$ is said to \df{cover} all pairs $(i, j) \in R \times C$.
The \df{cost} of a rectangle $(R, C)$ is defined as $|R| + |C|$.

Suppose a matrix $A$ is fixed; then a collection of rectangles is called
a \df{covering} of $A$ if for every $(i, j) \in \setm \times \setn$ there
exists a rectangle in the collection that covers $(i, j)$.
The \df{cost} of a collection is the sum of costs of all its rectangles.

Given a Boolean matrix $A \in \Bin^{m \times n}$,
the \df{cost} of $A$ is defined as the smallest cost of a covering of $A$.
It is not difficult to show that the cost of $A$ equals $\ortwo(A)$
as defined above.

Similarly, we can think of minimizing the \emph{size} of
a covering, i.e., the number of rectangles in a collection
instead of their total cost. The smallest size of a covering of $A$
is called the \df{\OR-rank} (or the \df{Boolean rank}) of $A$, denoted $\rankor A$.

\section{Fractional and greedy coverings}
\label{s:techniques}

In the rest of the paper we interpret
the covering problems for Boolean matrices as special cases
of the general set cover problem.
In this section we recall this general setting
and present two main techniques that we apply:
linear programming duality and greedy heuristics.

An instance of the \df{(weighted) set cover} problem consists of
a set $U$, a family of its subsets, $\Family \sset 2^U$, and
a weight function, which is a mapping $w \colon \Family \to \N$.
Every set $S \in \Family$ is said to \df{cover} all elements
$s \in S \sset U$.
The goal is to find a subfamily $\SubFamily \sset \Family$
that is a \df{covering} (i.e., it covers all elements from $U$:
$\bigcup_{S \in \SubFamily} S = U$)
and has the smallest possible total weight (i.e., it minimizes the functional
$\sum_{S \in \SubFamily} w(S)$ amongst all coverings).
In the \df{unweighted} version of the problem, $w(S) = 1$
for all $S \in \Family$, so the total weight of a covering
is just its \df{size} (number of elements in \SubFamily).
In both versions, \Family is usually assumed to be a feasible
solution, which means that every $s \in U$ belongs to at least one
set from \Family: that is, $\bigcup_{S \in \Family} S = U$.

It is instructive, throughout this section, to have
particular instances of the set cover problem in mind,
namely those of covering Boolean matrices with rectangles
as in Section~\ref{s:pre}.
In the following sections, we refer to them as \df{weighted} and \df{unweighted}
\df{set covering formulations};
their optimal solutions correspond to the values of $\ortwo(A)$ and $\rankor A$ respectively.

\subsubsection*{Fractional coverings}
The set cover problem can easily be recast as an integer program:
see Figure~\ref{f:prog:integer}.
For each $S \in \Family$, this program has an integer variable $x_S \in \Bin$:
the interpretation is that $x_S = 1$ if and only if $S \in \SubFamily$,
and the constraints require that every element is covered.
\df{Feasible} solutions are in a natural one-to-one correspondence
with coverings of $U$,
and the optimal value in the program is the smallest weight of a covering.

\begin{figure}
\subfloat[Integer program]{
\label{f:prog:integer}
\begin{minipage}{0.3\textwidth}
$\sum\limits_{S \in \Family} w(S)\,x_S \to \min$\\
$x_S \in \Bin \text{\ for all $S \in \Family$}$\\
$\sum\limits_{\substack{S \in \Family \colon\\u \in S}} x_S \ge 1 \text{\ for all $u \in U$}$
\end{minipage}
}
\subfloat[Linear relaxation]{
\label{f:prog:linear}
\begin{minipage}{0.3\textwidth}
$\sum\limits_{S \in \Family} w(S)\,x_S \to \min$\\
$0 \le x_S \le 1 \text{\ for all $S \in \Family$}$\\
$\sum\limits_{\substack{S \in \Family \colon\\u \in S}} x_S \ge 1 \text{\ for all $u \in U$}$
\end{minipage}
}
\subfloat[Dual of the linear relaxation]{
\label{f:prog:dual}
\begin{minipage}{0.35\textwidth}
$\sum\limits_{u \in U} y_u \to \max$\\
$y_u \ge 0 \text{\ for all $u \in U$}$\\
$\sum\limits_{u \in S} y_u \le w(S) \text{\ for all $S \in \Family$}
\vphantom{\sum\limits_{\substack{S \in \Family \colon\\u \in S}}}$
\end{minipage}
}
\caption{Integer and linear programs for the set cover problem}
\label{f:prog}
\end{figure}

The \df{linear programming relaxation} of this integer program
is obtained by interpreting variables $x_S$ over reals:
see Figure~\ref{f:prog:linear}.
Now $0 \le x_S \le 1$ for each $S \in \Family$.
Feasible solutions to this program are called \df{fractional coverings}.
Suppose the optimal cost in the original set cover problem is \OptCover.
Then the integer program in Figure~\ref{f:prog:integer} has optimal value \OptCover,
and its relaxation in Figure~\ref{f:prog:linear} optimal value $\OptFrac \le \OptCover$.

Finally, define the \df{dual} of this linear program: this is also
a linear program, and it has a (real) variable $y_u$
for each element $u \in U$;
see Figure~\ref{f:prog:dual}.
This is a maximization problem, and
its optimal value coincides with \OptFrac
by the strong duality theorem.

The following lemma summarizes the properties of these programs
needed for the sequel.

\begin{lemma}
\label{l:dual}
If $(y_u)_{u \in U}$ is a feasible solution to the dual, then
$\sum_{u \in U} y_u \le \OptFrac \le \OptCover$.
There exists a feasible solution to the dual, $(y^*_u)_{u \in U}$,
such that $\sum_{u \in U} y^*_u = \OptFrac$.
\end{lemma}

The proof can be found in, e.g.,~\cite{Karloff}.
We use the first part of Lemma~\ref{l:dual} in Section~\ref{s:lt}
to obtain a lower bound on \OptCover
and the second part in Section~\ref{s:kronecker}
to associate ``weights'' with $1$-elements in the matrix.

\subsubsection*{Greedy coverings}
%
The greedy heuristic for the unweighted set cover problem
works as follows. It maintains the set of uncovered elements,
initially $U$, and iteratively adds to \SubFamily (which is initially empty)
a set $S \in \Family$ which covers the largest number
of yet-uncovered elements.
Any covering obtained by this (nondeterministic) procedure is called
a \df{greedy covering}.
(There is a natural extension to the weighted version as well.)

A standard analysis of the greedy heuristic is performed in the
framework of approximation algorithms: the size of a greedy covering
is at most $O(\log|U|)$ times larger than that of the optimal covering~\cite{ch79,lovasz}.
But for our purposes a different upper bound will be more convenient:
an ``absolute'' upper bound in terms of the ``density'' of the instance.
Such a bound is given by the following result,
which is substantially less well-known:

\begin{lemma}[greedy covering lemma]
\label{l:greedy}
Suppose every element $s \in U$ is contained in at least $\gamma |\Family|$
sets from \Family, where $0 < \gamma \le 1$.
Then the size of any greedy covering does not exceed
\begin{equation*}
\left\lceil
\frac{1}{\gamma}
\ln^+ (\gamma |U|)
\right\rceil
+
\frac{1}{\gamma},
\end{equation*}
where $\ln^+(x) = \max(0, \ln x)$ and\, $\ln x$ is the natural logarithm.
\end{lemma}

Several versions of the lemma
can be found in the literature.
It was proved for the first time in 1972
by Sapozhenko~\cite{sap72} and appears
in later textbooks~\cite[Lemma~9 in Section~3, pp.~136--137]{dmimvk},
\cite[pp.~134--135]{dmimvk-de}.
A slightly different form, attributed to Stein~\cite{stein} and Lov\'asz~\cite{lovaszDM},
was independently obtained later and
is sometimes known as the Lov\'asz--Stein theorem;
yet another proof is due to Karpinski and Zelikovsky~\cite{kz96}.
Recent treatments with applications and more detailed discussion
can be found in Deng~et~al.~\cite{deng11} and in Jukna's textbook~\cite[pp.~34--37]{j-comb}.

Since the upper bound of Lemma~\ref{l:greedy} is hardly a standard
tool in theoretical computer science as of now, a remark on the proof is in order.
A standalone proof goes via the following fact:
on each step of
the greedy algorithm the number of yet-uncovered elements shrinks
by a constant factor, determined by the density parameter $\gamma$
and the size of the instance.
Alternatively, one can use the result due to Lov\'asz~\cite{lovaszDM}
that the size of any greedy covering is within a factor of $1 + \log |U|$
from the optimal \emph{fractional} covering.
Since assigning the value
$(\min_{s \in U} |\{ S \in \Family \colon s \in S\}|)^{-1} = 1 / \gamma |U|$
to all $x_S$, $S \in \Family$, in the linear program in Figure~\ref{f:prog:linear}
leads to a feasible solution, an upper bound of $(1 / \gamma) \cdot (1 + \log |U|)$ follows.

We use Lemma~\ref{l:greedy} in Section~\ref{s:ks} to obtain an upper bound
on the $\ortwo$-complexity of Kneser-Sierpi\'nski matrices.
We remark that instead of greedy coverings one can use random coverings
to essentially the same effect (cf. Deng~et~al.~\cite{deng11}).

\section{Lower bound for the full triangular matrices}
\label{s:lt}

Define the $n \times n$ \df{full triangular matrix}
$T_n = (t_{i j})_{0 \le i, j < n}$ by $t_{i j} = 1$ if $i < j$ and $t_{i j} = 0$ otherwise.
This matrix $T_n$ is the adjacency matrix of the Hasse diagram of
the strict linear order $0 < 1 < \cdots < n - 1$;
it has $1$s everywhere above the main diagonal and $0$s on the diagonal
and below.
In this section, we study the smallest size of depth-$2$ rectifier networks
that express $T_n$.

Define
$s(n) = n(\lfloor \log_2 n \rfloor + 2) - 2^{\lfloor \log_2 n \rfloor + 1}$
for $n \geq 1$.
Note that $s(n)$ is the so-called binary entropy
function, sequence \seqnum{A003314} in Sloane's {\it Encyclopedia
of Integer Sequences} \cite{Sloane}.  Its properties were studied
previously by Morris \cite{Morris} because of its connection with
mergesort.

\begin{theorem}
\label{th:fft}
All fractional coverings of $T_n$ have cost of at least $s(n)$.
\end{theorem}

\begin{corollary}
\label{cor:ft}
$\ortwo(T_n) = s(n)$.
\end{corollary}

Note that the equality of Corollary~\ref{cor:ft} gives the exact value of $\ortwo(T_n)$.
The upper bound is an easy divide-and-conquer argument
(reproduced in Appendix for completeness), and the main
challenge is to obtain the lower bound.

Consider the weighted set covering formulation for $T_n$, where
the optimal value is $\ortwo(T_n)$
as discussed in Section~\ref{s:techniques}.
By Lemma~\ref{l:dual},
it suffices to find a feasible solution to the dual linear program
with the value $s(n)$.
Our feasible solution is given by a certain infinite diagonal matrix
$M$, with rows and columns indexed by the natural numbers,
defined as follows:
$$ M_{i,j} = \begin{cases}
	2, & \text{if $j-i = 1$}; \\
	1, & \text{if $j-i = 2^q$ for some $q \geq 1$}; \\
	0, & \text{otherwise.}
	\end{cases}
$$
The first 17 rows and columns of $M$ are displayed in Figure~\ref{f:lt:matrix}.
Notice that each row is a shift, by $1$, of the preceding row.

\begin{figure}
\subfloat[Portion of the matrix $M$]{
\label{f:lt:matrix}
\newcommand*{\hcell}[1]{\rlap{\tiny #1}\phantom{0}}
\raisebox{\depth}{
$
\begin{smallmatrix}
i\backslash j
& \hcell{0} & \hcell{1} & \hcell{2} & \hcell{3}
& \hcell{4} & \hcell{5} & \hcell{6} & \hcell{7}
& \hcell{8} & \hcell{9} & \hcell{10} & \hcell{11}
& \hcell{12} & \hcell{13} & \hcell{14} & \hcell{15}
& \hcell{16}  \\
\hline
\vphantom{1^1}
\text{\tiny 0 }& 0 & 2 & 1 & 0 & 1 & 0 & 0 & 0 & 1 & 0 & 0 & 0 & 0 & 0 & 0 & 0 & 1 \\
\text{\tiny 1 }& 0 & 0 & 2 & 1 & 0 & 1 & 0 & 0 & 0 & 1 & 0 & 0 & 0 & 0 & 0 & 0 & 0 \\
\text{\tiny 2 }& 0 & 0 & 0 & 2 & 1 & 0 & 1 & 0 & 0 & 0 & 1 & 0 & 0 & 0 & 0 & 0 & 0 \\
\text{\tiny 3 }& 0 & 0 & 0 & 0 & 2 & 1 & 0 & 1 & 0 & 0 & 0 & 1 & 0 & 0 & 0 & 0 & 0 \\
\text{\tiny 4 }& 0 & 0 & 0 & 0 & 0 & 2 & 1 & 0 & 1 & 0 & 0 & 0 & 1 & 0 & 0 & 0 & 0 \\
\text{\tiny 5 }& 0 & 0 & 0 & 0 & 0 & 0 & 2 & 1 & 0 & 1 & 0 & 0 & 0 & 1 & 0 & 0 & 0 \\
\text{\tiny 6 }& 0 & 0 & 0 & 0 & 0 & 0 & 0 & 2 & 1 & 0 & 1 & 0 & 0 & 0 & 1 & 0 & 0 \\
\text{\tiny 7 }& 0 & 0 & 0 & 0 & 0 & 0 & 0 & 0 & 2 & 1 & 0 & 1 & 0 & 0 & 0 & 1 & 0 \\
\text{\tiny 8 }& 0 & 0 & 0 & 0 & 0 & 0 & 0 & 0 & 0 & 2 & 1 & 0 & 1 & 0 & 0 & 0 & 1 \\
\text{\tiny 9 }& 0 & 0 & 0 & 0 & 0 & 0 & 0 & 0 & 0 & 0 & 2 & 1 & 0 & 1 & 0 & 0 & 0 \\
\text{\tiny 10}& 0 & 0 & 0 & 0 & 0 & 0 & 0 & 0 & 0 & 0 & 0 & 2 & 1 & 0 & 1 & 0 & 0 \\
\text{\tiny 11}& 0 & 0 & 0 & 0 & 0 & 0 & 0 & 0 & 0 & 0 & 0 & 0 & 2 & 1 & 0 & 1 & 0 \\
\text{\tiny 12}& 0 & 0 & 0 & 0 & 0 & 0 & 0 & 0 & 0 & 0 & 0 & 0 & 0 & 2 & 1 & 0 & 1 \\
\text{\tiny 13}& 0 & 0 & 0 & 0 & 0 & 0 & 0 & 0 & 0 & 0 & 0 & 0 & 0 & 0 & 2 & 1 & 0 \\
\text{\tiny 14}& 0 & 0 & 0 & 0 & 0 & 0 & 0 & 0 & 0 & 0 & 0 & 0 & 0 & 0 & 0 & 2 & 1 \\
\text{\tiny 15}& 0 & 0 & 0 & 0 & 0 & 0 & 0 & 0 & 0 & 0 & 0 & 0 & 0 & 0 & 0 & 0 & 2 \\
\text{\tiny 16}& 0 & 0 & 0 & 0 & 0 & 0 & 0 & 0 & 0 & 0 & 0 & 0 & 0 & 0 & 0 & 0 & 0
\end{smallmatrix}
$\qquad
}
}
\subfloat[Definition of $a,b,c,d,e$]{
\label{f:lt:definition}
\includegraphics[width=0.5\textwidth]{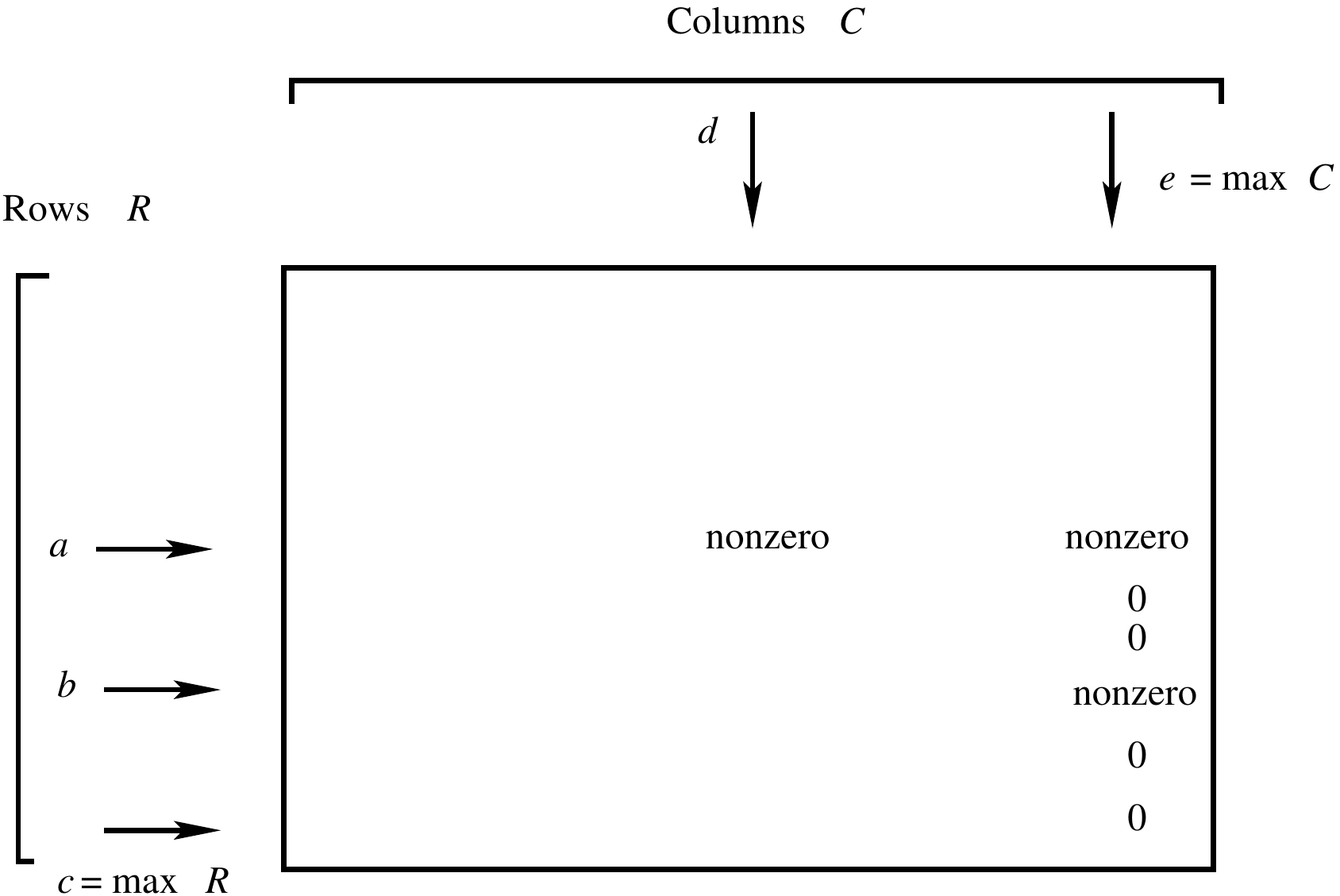}
}
\caption{Illustrations for the proof of Theorem~\ref{th:fft}}
\end{figure}

\begin{lemma}
\label{l:value}
The sum of the elements of $M^{(n)}$,
the $n \times n$ upper left submatrix of $M$,
is equal to $s(n)$.
\end{lemma}

\begin{lemma}
\label{l:feasible}
$y_{i,j} = M_{i,j}$ for $0 \leq i < j < n$
is a feasible solution to the dual program.
\end{lemma}

\begin{proof}[Proof of Lemma~\ref{l:value}]
$M^{(n+1)}$ is obtained from $M^{(n)}$ by concatenating a row of
$0$'s on the bottom, and a column that contains a single $2$
and $1$'s corresponding to the powers of $2$ that are $\leq n$.
In other words, $s(n+1) = s(n) + \lfloor \log_2 n \rfloor  + 2$.
The result now follows by an easy induction.
\end{proof}

\begin{proof}[Proof of Lemma~\ref{l:feasible}]
To prove feasibility, we need to see that for each pair of nonempty sets
$R, C \subseteq \{ 0, 1, \ldots, n - 1 \}$ 
with $\max R < \min C$---only such pairs $(R, C)$ are rectangles of $T_n$---we have
\begin{equation}
\sum_{{i \in R} \atop{j \in C}} M_{i,j} \leq |R|+|C|.
\label{ineq}
\end{equation}
Here $R$ corresponds to
a choice of rows of $M$ and $C$ to a choice of columns.

Suppose there exists a counterexample to \eqref{ineq}.
Among all counterexamples to \eqref{ineq}, consider one with the
smallest possible value of $|R|+|C|$.  If $|R|=1$ then
since at most one entry in each row is $2$ and all others are
either $0$ or $1$, we clearly have 
$\sum_{{i \in R} \atop{j \in C}} M_{i,j} \leq |R|+|C| = |C|+1$.
Hence $|R| \geq 2$.  The same argument applies if $|C|=1$.  Thus
the minimal counterexample to $\eqref{ineq}$ has at least two rows
and columns.

We now observe that the row sum of each row in our counterexample
is at least $2$.  For if it is $0$ or $1$ we could omit that row,
and~\eqref{ineq} would still be violated.   The same argument
applies to the column sums.
We now prove
\begin{claim*}
Suppose there are at least two nonzero elements in
the submatrix of $M$ formed by
rows $0, 1, \ldots, b$ and column $e$ of $M$.  Then $e \leq 2b$.
\end{claim*}

\begin{proof}
The nonzero elements in column $e$ occur precisely in the rows
numbered $e-1, e-2, \ldots, e-2^i$ where $i$ is the largest
integer with $e - 2^i \geq 0$.  So if there are nonzero elements
in rows $0, 1, \ldots, b$, these would be given by $e - 2^i$
and $e - 2^{i-1}$.  So $e - 2^{i-1} \leq b$.  
It now follows that $e \leq b + 2^{i-1} = b + {1 \over 2} \cdot 2^i
\leq b + {1 \over 2} e$ (since $e \ge 2^i$), and so $e \leq 2b$.
This concludes the proof of the claim.
\end{proof}

Now let us assume that our minimal counterexample has
$c = \max R$.  Let $e = \max C$.  Since column $e$ has $2$
nonzero elements, by the Claim above we know $e \leq 2c$.
Now let $b$ be the largest element $\leq c$ in $R$ 
for which there is a nonzero element in column $e$; this must
exist since column $e$ has at least two nonzero elements.
Let $a$ be any row $< b$ in $R$ with a nonzero element in column $e$.
Again, this must exist since column $e$ has at least two
nonzero elements.   Finally, let $d$ be any column $< e$ in $C$ with
a nonzero element in row $a$.  This must exist because every row
in $R$ has at least two nonzero elements.
We claim $d \leq c$.

To see this, note that $b = e-2^j \leq c$ for some $j \geq 0$.  (In fact,
$j = \lceil \log_2 (e-c) \rceil$.)
Then we must have $a = e-2^k \geq 0$ where $k \geq j+1$.
Then $d - a = 2^\ell$ for some $\ell$.  
So $d - a = d-(e-2^k) = 2^\ell$ and hence
$d = e + 2^\ell- 2^k$.  Since $d < e$ we have $\ell < k$.  
So $d \leq e + 2^{k-1} - 2^k = e - 2^{k-1} \leq e - 2^j = b \leq c$.
This is illustrated in Figure~\ref{f:lt:definition}.

\begin{figure}
\end{figure}

Now $\max R < \min C$, but $d \leq c$ while $d\in C$ and $c \in R$,
a contradiction.
Hence there are no minimal counterexamples and no counterexamples
at all.  Thus \eqref{ineq} holds.
It follows that $M$ represents a feasible solution.
This concludes the proof of Lemma~\ref{l:feasible}.
\end{proof}

Let us complete the proof of Theorem~\ref{th:fft}.
Apply the first part of Lemma~\ref{l:dual} to the weighted set covering formulation of
the problem and take the solution $y_{i,j} = M_{i,j}$, $0 \leq i < j < n$,
as described above.
This solution has value $s(n)$ by Lemma~\ref{l:value} and is feasible by
Lemma~\ref{l:feasible}.
Hence, all fractional coverings have cost at least $s(n)$.

\section{Upper bound for Kneser-Sierpi\'nski matrices}
\label{s:ks}

Suppose $n = 2^k$.
A \df{Kneser-Sierpi\'nski matrix} (or a \df{disjointness matrix})
of size $2^k \times 2^k$ is
the matrix $D_n$ defined as follows.
Rows and columns of the matrix are indexed from $0$ to $2^k - 1$.
The matrix has a $1$ at all positions $(i, j)$
such that $i$ and $j$ have no common $1$ in their binary expansion;
all other elements of the matrix are $0$.

Note that if we identify each number from \kbits
with a subset of \setk in the natural way, then $D_n$
is naturally associated with a Boolean function that
maps a pair of subsets of \setk to $1$ if they are disjoint,
and to $0$ if they have an element in common.
An alternative way to define $D_n$ is by a recurrence
$D_{2 n} = \left(\begin{smallmatrix} D_n & D_n \\ D_n & 0 \end{smallmatrix}\right)$
for $n \ge 1$; $D_1 = (1)$; here subsets of \setk are
ordered lexicographically. Using the antilexicographic order
for rows and the lexicographic order for columns would lead
to a lower triangular matrix.

What is the size of smallest depth-$2$ rectifier networks
that express Kneser-Sierpi\'nski matrices?
Jukna and Sergeev~\cite[Lemma~4.2]{js13} prove that
\begin{equation}
\label{eq:ks-js}
n^{\frac{1}{2} \log 5} / \polylog(n)
\le
\ortwo(D_n)
\le
n^{\log(1 + \sqrt{2})} \cdot \polylog(n),
\end{equation}
and in this section, we prove the following result:

\begin{theorem}
\label{th:ks}
$\ortwo(D_n) \le n^{\log(9/4)} \cdot \polylog(n)$.
\end{theorem}

Note that $\frac{1}{2} \log 5 \approx 1.16096$, $\log(9/4) \approx 1.16993$,
and $\log(1 + \sqrt{2}) \approx 1.27$.




Suppose $n = 2^k$ as above, and
let $D^{x, y}_{[k]}$ be the submatrix of $D_n$ whose rows and columns
correspond to $x$-sized and $y$-sized subsets of \setk, respectively.
This matrix $D^{x, y}_{[k]}$ has size $\binom{k}{x} \times \binom{k}{y}$.
If $x = y$, then $D^{x, x}_{[k]}$ is the adjacency matrix
of the Kneser graph~\cite{on-kneser}.

For $0 \le y \le x \le k$, write $z = (k - x - y) / 2$ and
$f(x, y) = \binom{k}{x, z, k - x - z} / \binom{2 z}{z}$.%
\footnote{We use the standard notation for multinomial coefficients:
$\binom{k}{a, b, c} = \frac{k!}{a!\, b!\, c!}$ provided that $a + b + c = k$.}
Jukna and Sergeev~\cite[Lemma~4.2]{js13} show that
all coverings of $D^{x, x}_{[k]}$ have cost at least $f(x, x) / \poly(k)$,
and this gives the lower bound in equation~\eqref{eq:ks-js}:
taking $x = 0.4 k$ brings $f(x, x)$ to its maximum of $n^{\frac{1}{2} \log 5}$,
if we disregard factors polylogarithmic in~$n = 2^k$.
Our Theorem~\ref{th:ks}
follows from Lemmas~\ref{l:drop-polylog} and~\ref{l:opt-ub}
below.


\begin{lemma}
\label{l:drop-polylog}
There exists a covering of $D^{x, y}_{[k]}$
with cost at most $f(x, y) \cdot \poly(k)$.
\end{lemma}

\begin{proof}
Consider \Family, the family of all \df{ordered bipartitions} of $\{1, \ldots, k\}$
into sets of size $x + z$ and $y + z$, where $z = (k - x - y) / 2$.
Technically, an ordered bipartition is simply a subset of \setk,
but it is more instructive to view it as an ordered pair:
this subset and its complement.
Every such bipartition, $(S, \compl S)$, corresponds to a (maximal) rectangle in $D^{x, y}_{[k]}$;
elements of $D^{x, y}_{[k]}$ covered by the rectangle are pairs $(X, Y)$ of disjoint
sets that \df{respect} the bipartition: $X \sset S$ and $Y \sset \compl S$.
\par
Use the greedy covering lemma (Lemma~\ref{l:greedy})
for the unweighted set covering formulation
with \Family.
There are $\binom{k}{x + z}$ bipartitions in this family,
and every pair of disjoint sets $(X, Y)$ of size $x$ and $y$
respects $\binom{2 z}{z}$ of them,
so $\gamma = \binom{2 z}{z} / \binom{k}{x + z}$ and
any greedy covering will contain at most $N$~sets, where
\begin{equation*}
N =
\frac{\binom{k}{x + z}}{\binom{2 z}{z}} \cdot
(1 + \ln(4^k)) + 1 =
\frac{\binom{k}{x + z}}{\binom{2 z}{z}} \cdot \poly(k).
\end{equation*}
For every bipartition in the covering, the corresponding $1$-rectangle
in $D^{x, y}_{[k]}$ will include $\binom{x + z}{z}$ rows and $\binom{y + z}{z}$
columns; its cost will be at most $2\,\binom{x + z}{z}$ as $y \le x$.
So the total cost of the covering will not exceed
\begin{equation*}
\textstyle\binom{x + z}{z} \cdot 2 N
=
\frac{2\, \binom{k}{x + z} \binom{x + z}{z}\cdot \poly(k)}{\binom{2 z}{z}}
=
\frac{\binom{k}{x, z, k - x - z} \cdot \poly(k)}{\binom{2 z}{z}} = f(x, y) \cdot \poly(k).
\qedhere
\end{equation*}
\end{proof}

\begin{corollary}
\label{cor:kneser}
Suppose $0 \le m \le k/2$ and let $D^m_k = D^{m, m}_{[k]}$ be the adjacency matrix
of the (bipartite) Kneser graph: vertices in each part are size-$m$
subsets of $\{1, \ldots, k\}$, and two vertices from different parts are adjacent
if and only if the subsets are disjoint.
Then $d(m, k) / \poly(k) \le \ortwo(D^m_k) \le d(m, k) \cdot \poly(k)$ where
$d(m, k) = \binom{k}{m, k/2 - m, k/2} / \binom{k - 2 m}{k/2 - m}$.
\end{corollary}

\begin{lemma}
\label{l:opt-ub}
If $0 \le y \le x \le k$, then $f(x, y) \le 2^{k \log(9/4)} \cdot \poly(k)$,
and there exists a pair $(x^*, y^*)$ such that $f(x^*, y^*) \ge 2^{k \log(9/4)} / \poly(k)$.
\end{lemma}

\begin{proof}
As above, let $2 z = k - (x + y)$.
Denote $\alpha = z / k$ and
recall that the values of the binomial coefficients
may be estimated with the help of the binary entropy function
(not to be confused with $s(n)$ from Section~\ref{s:lt}, also known under this name):
$\binom{k}{\lambda k} \sim 2^{H(\lambda) k + O(\log k)}$ as $k \to \infty$,
where $H(\lambda) = - \lambda \log \lambda - (1 - \lambda) \log(1 - \lambda)$.
This formula follows from Stirling's approximation
for the factorial~\cite[Chapter~9 and Solution to Exercise~9.42]{concrete-math}.
Now
\begin{equation*}
f(x, y) = \frac{ \binom{k}{z} \binom{k - z}{x} }{ \binom{2 z}{z} } \le
          \frac{ \binom{k}{z} \binom{k - z}{(k - z) / 2} }{ \binom{2 z}{z} } =
          \frac{ 2^{k H(\alpha)} 2^{(1 - \alpha) k H(1/2)} }{ 2^{2 \alpha k H(1/2)} }
          \cdot \poly(k)
        = 2^{(H(\alpha) + 1 - 3 \alpha) k} \cdot \poly(k)
\end{equation*}
as $H(1/2) = 1$.
Simple calculations show that for $0 < \alpha < 1/2$ the inequality
$H(\alpha) + 1 - 3 \alpha \le H(1/9) + 1 - 3 \cdot 1/9 = \log(9/4)$ holds.
This corresponds to $x = 4 / 9 \cdot k$ and $y = 3 / 9 \cdot k$.
\end{proof}

To complete the proof of Theorem~\ref{th:ks},
it remains to note that a union of coverings of matrices
$D^{x, y}_{[k]}$ for all pairs $x, y$ with $0 \le x, y \le k$
constitutes a covering of $D_n$. For $0 \le y \le x \le k$,
the coverings are constructed by Lemma~\ref{l:drop-polylog},
and for $x \le y$ the construction just swaps the roles of $x$ and $y$.
Since there are only $(k + 1)^2 = \polylog(n)$
pairs $x, y$ in total, the desired follows from Lemma~\ref{l:opt-ub}.

\begin{myremark}
Although Theorem~\ref{th:ks} leaves a gap
between the bounds on $\ortwo(D_n)$,
the greedy strategy is, in fact, optimal:
For each $D^{x, y}_{[k]}$, it suffices to
use bipartitions into sets of size $\ell$ and $k - \ell$,
for some $\ell = \ell(k; x, y)$.
(See Appendix for more details.)
Our choice of $\ell$ in Lemma~\ref{l:drop-polylog} is $\ell = x + (k - x - y) / 2$,
and the optimal choice, $\ell = \ell^*(k; x, y)$,
will deliver a tight upper bound on $\ortwo(D_n)$.
Numerical experiments seem to indicate that
the actual value of $\ortwo(D_n)$ is within a $\polylog(n)$ factor
from $n^{\frac{1}{2} \log 5}$, but no formal proof is known to us.
\end{myremark}


\section{Lower bound for Kronecker products}
\label{s:kronecker}

Given two matrices $K\in\{0,1\}^{m_1\times n_1}$  and
$M\in\{0,1\}^{m_2\times n_2}$, their \df{Kronecker}
(or \df{tensor}) \df{product} is the Boolean matrix
$K\otimes M$ of size $(m_1\cdot m_2) \times (n_1\cdot n_2)$
defined as follows.
Its rows are indexed by pairs $(i_1, i_2)$ and its columns
by pairs $(j_1, j_2)$ where $1 \le i_s \le m_s$ and $1 \le j_s \le n_s$
for $s = 1, 2$.
The entry of $K \otimes M$ at position $((i_1, i_2), (j_1, j_2))$
is defined as $K_{i_1,j_1}\cdot M_{i_2,j_2}$.

In this section we prove a lower bound on the $\orany(\cdot)$-measure of Kronecker products.
%
%
%
%
Recall that the Boolean rank $\orrank(K)$ is the optimal value of the unweighted
set covering formulation (as in Figure~\ref{f:prog:integer}) where
the set of $1$-entries in the matrix $K$ is covered by all-$1$ rectangles.
In the linear
relaxation of this problem (as in Figure~\ref{f:prog:linear}), the goal
is to assign weights $w(R)\in[0,1]$ to
each $1$-rectangle $R$ such that $\sum_{(i,j)\in R}w(R)\geq 1$ for each
$1$-entry $(i,j)$ of $K$, minimizing $\sum w(R)$.
Let \df{the fractional rank}
$\frrank(K)$ be the optimal value of this linear relaxation.
The integrality gap result for the set cover problem~\cite{lovaszDM}
and the duality theorem
imply that $\orrank(K) / (1 + \log m_1 n_1) \le \frrank(K) \le \orrank(K)$.
In the graph-theoretic language, the number $\frrank(K)$ is
the \emph{fractional biclique cover number}, denoted by $bc^*(G)$
where $K$ is the adjacency matrix of the (bipartite) graph $G$.
Fractional rank is known to be bounded from below by
the fooling set number, see Watts~\cite[Theorem~2.2]{Watts06}.

\begin{theorem}
\label{th:dirprod}
For any pair $K$, $M$ of Boolean matrices,
$\orany(K\otimes M)\geq \frrank(K)\cdot\orany(M)$.
\end{theorem}


\begin{proof}
First consider the unweighted set covering formulation for $K$,
where the optimal value is $\orrank(K)$
as discussed in Section~\ref{s:techniques},
and take its linear relaxation, with the optimal value $\frrank(K)$.
%
By Lemma~\ref{l:dual}, there is an assignment of weights to $1$-elements of
this matrix, $w(i, j)\in[0,1]$ for all $(i, j)$ with $K_{i, j} = 1$,
such that the following two conditions are satisfied (see Figure~\ref{f:prog:dual}).
First, for each $1$-rectangle $R \times C$ of $K$,
the sum $\sum_{(i,j)\in R \times C}w(i, j)$ is at most~$1$.
Second, $\sum_{(i, j):K_{i,j}=1}w(i, j) = \frrank(K)$.
%
%
%

\def\Left{{\mathrm{To}}}
\def\Right{{\mathrm{From}}}

Now let $\Net=(V,E,\sources,\sinks)$ be a rectifier network of size $\orany(K \otimes M)$
that expresses $Q=K\otimes M$,
where $K$ and $M$ have size as above.
For an edge $e\in E$, let
$\Left(e)\subseteq \{1,\ldots,m_1\}\times\{1,\ldots,m_2\}$
be the set of row indices $(i_1,i_2)$ of $Q$ such that the
node $\sinks(i_1,i_2)$ is reachable from the target of $e$.
Similarly, let
$\Right(e)\subseteq\{1,\ldots,n_1\}\times\{1,\ldots,n_2\}$
be the set of column indices $(j_1,j_2)$ of $Q$ such that
the source of $e$ is reachable from $\sources((j_1,j_2))$.
Then $R(e)=(\Left(e),\Right(e))$ is a rectangle of $Q$.
Moreover, 
define $\pi_s((i_1,i_2),(j_1,j_2))=(i_s,j_s)$ for $s=1,2$
and $\pi_s(R)=\{\pi_s(r,c) \colon (r,c)\in R\}$.
Then $\pi_1(R(e))$ and $\pi_2(R(e))$ are
rectangles in $K$ and $M$ respectively.

We assign real weights based on $w$ to each edge $e$ of $\Net$ by the following rule:
\[w'(e)=\mathop\sum\limits_{(i,j)\in \pi_1(R(e))}w(i,j).\]
Since $\pi_1(R(e))$ is a rectangle of $K$,
one of the constraints on $w$ ensures that $w'(e)\leq 1$ for each edge $e$
of $\Net$. Consequently, $\sum_{e\in E}w'(e)\leq|E| = \orany(K \otimes M)$;
furthermore, the following chain of inequalities holds:
\begin{align}
\orany(K \otimes M) & \ge \quad \sum_{e\in E}w'(e)\quad=\quad
\sum_{e\in E}\ \sum\limits_{(i_1,j_1)\in\pi_1(R(e))}w(i_1,j_1)\notag\\
&=\nvts\sum_{(i_1,j_1):K_{i_1,j_1}=1}\nvts
w(i_1,j_1)\cdot |\{e\in E:(i_1,j_1)\in\pi_1(R(e))\}|\notag\\
\label{eq-align-lp}
&=\nvts\sum_{(i_1,j_1):K_{i_1,j_1}=1}\nvts
w(i_1,j_1)\cdot|\{e\in E:
i_1\in\pi_1(\Left(e)),j_1\in\pi_1(\Right(e))\}|.
\end{align}
Fix an arbitrary entry $(i_1,j_1)$ of $K$ with $K_{i_1,j_1}=1$.
Consider the subgraph $\Net_{j_1\leadsto i_1}$ of
$\Net$ induced by the nodes that are reachable from some
source of the form $\sources(j_1,j_2)$ and from which a node
of the form $\sinks(i_1,i_2)$ is reachable---in other words, take
all nodes and edges on all paths from $\sources(j_1,j_2)$ to
$\sinks(i_1,i_2)$ for some $i_2$, $j_2$.
Then, since $K_{i_1,j_1}=1$, the node $\sinks(i_1,i_2)$ is
reachable from $\sources(j_1,j_2)$ in $\Net_{j_1\leadsto i_1}$ if
and only if $M_{i_2,j_2}=1$. So the network $\Net_{j_1\leadsto i_1}$
expresses $M$ (with the mappings $\sources'(j_2)=\sources(j_1,j_2)$
and $\sinks'(i_2)=\sinks(i_1,i_2)$).
Hence, the number of edges in $\Net_{j_1\leadsto i_1}$ is
at least $\orany(M)$. But by our definitions, the relations
$i_1\in\pi_1(\Left(e))$ and
$j_1\in\pi_1(\Right(e))$ hold together exactly for the edges $e$
of $\Net$ present in $\Net_{j_1\leadsto i_1}$.
Thus $|\{e\in E:i_1\in \pi_1(\Left(e)),j_1\in\pi_1(\Right(e))\}|\geq\orany(M)$
and we conclude from equation~\eqref{eq-align-lp}
%
that
\begin{equation*}
\orany(K \otimes M)
\geq \sum_{(i_1,j_1):K_{i_1,j_1}=1}w(i_1,j_1)\cdot\orany(M)
=\frrank(K)\cdot\orany(M).\qedhere
\end{equation*}
\end{proof}

\begin{myremark}
Let $\sumany(K)$ be the smallest size of an \df{unambiguous}
rectifier network that expresses $K$.
A rectifier network is unambiguous if for all $i$, $j$
it has at most one path from $\sources(j)$ to $\sinks(i)$.
Such networks are also known under the names of \SUM-circuits~\cite{js13}
and cancellation-free circuits~\cite{bf15}.
The same construction as above also proves the inequality
$\sumany(K\otimes M) \geq \frrank(K)\cdot\sumany(M)$.
\end{myremark}

\begin{corollary}
\label{cor:dirprod}
For any pair of matrices $K \in \Bin^{m_1 \times n_1}$
and $M \in \Bin^{m_2 \times n_2}$, and $\mathsf{L}\in\{\orany,\sumany\}$ it holds that
$\mathsf{L}(K\otimes M) \ge \orrank(K) \cdot \mathsf{L}(M) / (1 + \log m_1 n_1)$.
\end{corollary}

%
%

\emph{Acknowledgements.}
We are grateful to Stasys Jukna, Alexander Kulikov,
Igor Sergeev, and anonymous reviewers for comments and discussions.

\newpage

\bibliographystyle{plain}
\bibliography{fgc}

\newpage

\appendix

\section{Depth-3 lower bound in Example~\ref{our-example}}

Consider the matrix
$M_n=\left(\begin{array}{ll}1&1\\0&1\end{array}\right)\otimes J_n$ for some
$n\geq 1$ where $J_n$ is the $n\times n$ all-one matrix.
Known bounds give $\OR(M_n)\geq 4n+1$ and this bound is indeed attainable.
For $\OR_3(M_n)$, i.e. realization by some rectifier network of exact depth $3$
we show $\OR_3(M_n)=4n+3$ using the following lemma:
\begin{lemma}
Suppose $M$ is a Boolean matrix and $\Net = (V, E, \sources, \sinks)$ is
a rectifier network realizing $M$ of some depth $d$.
Then there exists a rectifier network $\Net'=(V,E',\sources,\sinks)$ with
$|\Net'|\leq|\Net|$ having depth at most $d$
satisfying the following conditions:
\begin{itemize}
\item[i)] whenever the $i_1$th and the $i_2$th row are the same in $M$, then the
sets $\{v\in V:(v,\sinks(i_1))\in E'\}$ and $\{v\in V:(v,\sinks(i_2))\in E'\}$
coincide;
\item[ii)] dually, whenever the $j_1$th and the $j_2$th column of $M$ are the same,
    then
  $\{v\in V:(\sources(j_1),v)\in E'\}=\{v\in V:(\sources(j_2),v)\in E'\}$.
\end{itemize}
\end{lemma}
\begin{proof}
Let $v=\sources(j)$ be a source node and let $X_j$ stand for
the set $\{w\in V:(v,w)\in E\}$ of its neighbours.
Since $\Net$ realizes $M$, the set of target nodes $\sinks(i)$ which are
reachable in $\Net$ is exactly the image under $\sinks$ of 
those indices $i$ for which $M_{i,j}=1$. Now for each column index $j$ let $j'$
be the index for which
the $j$th and the $j'$th column of $M$ is the same, $|X_{j'}|$ is the smallest
possible among these sets
and $j'$ is the smallest among these indices. Note that $j'$ is always
well-defined and whenever the $j_1$th and the $j_2$th column
coincide, then $j_1'=j_2'$.

Then, define $\Net_0$ as$ (V,E_0,\sources,\sinks)$ with
$E_0=E-\{(\sources(j),v)\}\cup\{(\sources(j),v):v\in X_{j'}\}$.
(That is, we reattach the edges coming out from sources to the neighbours of
 the representative source of their equivalence class.)

Then by the choice of the values $j'$ (in particular, with $|X_{j'}|$ having
been minimized) we have that i) is satisfied, $\Net_0$ also realizes $M$, the
depth is not increased (if $\Net$ is strictly levelled) and $|\Net'|\leq|\Net|$.
Applying the analogous transformation to the targets we get a network $\Net'$
satisfying ii) as well.
\end{proof}
Thus we get that there exists a depth-$3$ network of minimal size realizing $M_n$ such that
\begin{itemize}
\item each source $\sources(i)$ for $i=1,\ldots,n$ have the same set $X_1$ of neighbours;
\item each source $\sources(i)$ for $i=n+1,\ldots,2n$ have the same set $X_2$ of neighbours;
\item each target $\sinks(j)$ for $j=1,\ldots,n$ have the same set $Y_1$ of neighbours and
\item each target $\sinks(j)$ for $j=n+1,\ldots,2n$ have the same set $Y_2$ of neighbours
\end{itemize}
since the corresponding rows and columns coincide.
In this network there are $n(|X_1|+|X_2|+|Y_1|+|Y_2|)$ edges in total
between the outermost layers (and some additional edges between the two middle
layers. Clearly none of these sets can be empty (since all the rows and columns
are nonzero), and if any of them is a non-singleton set, the size of the
network is at least $5n>4n+3$. So in order to go below $5n$, $X_1=\{x_1\}$,
$X_2=\{x_2\}$ etc. have to be singleton sets. Now since not all rows (columns,
resp.) are equal, $x_1\neq x_2$ and $y_1\neq y_2$ has to hold, and there is
only one choice (because the sets are singletons) to wire the two middle layers
together, namely adding the edges $(x_1,y_1)$, $(x_1,y_2)$ and $(x_2,y_2)$,
giving $4n+3$ edges in total as optimal value for depth $d=3$.

Note that if the network is not required to be strictly levelled, we can merge
$x_1$ with $y_1$ and $x_2$ with $y_1$ and add only the edge $(x_1,x_2)$
reaching the optimal bound $4n+1$.

\section{Upper bound in Corollary~\ref{cor:ft}}

Recall that a \df{$\SUM$-circuit} for a matrix $M$
is the same as an \df{unambiguous} rectifier network:
it is a rectifier network that has
at most one path between any input---output pair.
The smallest size of an unambiguous rectifier network
that expresses $M$ is denoted by $\sumany(M)$;
similarly, $\sumtwo(M)$ is
the smallest size of an unambiguous rectifier network
of depth~$2$ that expresses $M$.
In the same way as rectifier networks of depth~$2$
correspond to rectangle coverings, \emph{unambiguous} rectifier networks
of depth~$2$ correspond
to rectangle \df{partitions} (that is, coverings with no overlap between rectangles).
If one views the matrices
as adjacency matrices of bipartite graphs, then the measures
$\ortwo(\cdot)$ and $\sumtwo(\cdot)$ correspond to minimal biclique coverings and
minimal biclique partitions, respectively.
Clearly, $\orany(M)\leq\sumany(M)$ and $\orparam{d}(M)\leq\sumparam{d}(M)$ for each depth $d$.
Also, if $M=\left(\begin{smallmatrix}M_1&M_2\\M_3&M_4\end{smallmatrix}\right)$,
then $\sumany(M)\leq\sum_{i=1}^4\sumany(M_i)$.

We show below that $\sumtwo(T_n) \le s(n) =
n(\lfloor \log_2 n \rfloor + 2) - 2^{\lfloor \log_2 n \rfloor + 1}$.
Theorem~\ref{th:fft} will then imply that $\ortwo(T_n) = \sumtwo(T_n) = s(n)$.

First, let $J_n$ be the $n\times n$ all-$1$ matrix and $J_{m,k}$
the $m\times k$ all-$1$ matrix. Clearly, $\sumtwo(J_{m,k})$ is $m+k$.
Second, observe that
$T_{2n}=\left(\begin{smallmatrix}T_n&J_n\\0&T_n\end{smallmatrix}\right)$ and
$T_{2n+1}=\left(\begin{smallmatrix}T_n&J_{n,n+1}\\0&T_{n+1}\end{smallmatrix}\right)$.
It follows that $\sumtwo(T_{2n})\leq 2\sumtwo(T_n)+2n$ and
$\sumtwo(T_{2n+1})\leq \sumtwo(T_n)+\sumtwo(T_{n+1})+2n+1$.
This shows, by induction, that $\sumtwo(T_n)\leq s(n)$, since
the induction basis is easily checked.

\section{Optimality of the greedy strategy for Kneser-Sierpi\'nski matrices}

Although Theorem~\ref{th:ks} leaves a gap
between the bounds of $\Omega(n^{1.16})$ and $O(n^{1.17})$ on $\ortwo(D_n)$,
the greedy strategy is, in fact, optimal.
We first give a brief sketch of the argument,
and then fill in all the details below.

Consider the linear relaxation of the set covering formulation
for each $D^{x, y}_{[k]}$.
Note that only maximal rectangles (i.e., those associated with bipartitions)
can participate in optimal fractional coverings.
In fact,
for any $\ell \in [ x, k - y ]$
there exists a fractional covering $\eta(\ell)$ of $D^{x, y}_{[k]}$
which uses only bipartitions into sets of size $\ell$ and $k - \ell$
and for which all ``covering'' constraints in the LP are tight;
it suffices to pick a single $\ell$ since
this fractional covering $\eta(\ell)$ uses
\emph{all} such bipartitions with multiplicity $1 / \binom{k - (x + y)}{\ell - x}$.
Hence, the problem reduces to an unweighted set covering formulation,
where the greedy heuristic achieves a value
within a factor of $1 + \log \binom{k}{x} \binom{k}{y} \le 1 + 2 k = \polylog(n)$ of
the optimum.

In more detail,
first consider an arbitrary weighted set cover problem:
let $S_1,\ldots,S_k\subseteq U$ be the sets, with $w_i>0$ being the cost of $S_i$.
Let $\mu=\min\{\frac{w_i}{|S_i|}:i=1,\ldots,k\}$ be the best cost/utility ratio offered by the sets.
Then, in the dual formulation of its LP relaxation, if one assigns uniformly $\mu$ to each element $u\in U$
of the universe, then each set $S_i$ gets $\mu\cdot|S_i|\leq w_i$ total charge, hence this uniform
distribution is a solution to the dual, hence $\mu\cdot|U|$ is a lower bound for the optimum of the
primal problem by the weak duality theorem.

For the case of the weighted covering by rectangles, a rectangle of size $k\times m$ has
cost $k+m$ and covers $km$ elements, hence its offered ratio is $\frac{k+m}{km}=\frac{1}{k}+\frac{1}{m}$,
i.e. it decreases strictly by increasing either $k$ or $m$, thus the best ratios are always offered by maximal rectangles.

Now considering a rectangle $R$ in a matrix $D^{x,y}_{[k]}$, formed by the rows $X_1,X_2,\ldots,X_k$ and
columns $Y_1,\ldots,Y_m$ we have by definition that each $X_i$ is disjoint from each $Y_j$, thus
choosing $S=\mathop\bigcup\limits_{i=1}^kX_i$ we have that $R$ is a subrectangle of the rectangle
corresponding to the bipartition $(S,\overline{S})$, yielding that only rectangles corresponding to
bipartitions can be maximal. On the other hand, any such rectangle is clearly maximal.
Denoting $|S|$ by $\ell$ we get that the ratio offered by these rectangles is $\mu(k,x,y,\ell)=\frac{1}{\binom{\ell}{x}}+\frac{1}{\binom{k-\ell}{y}}$.
Then setting $\ell^*=\ell^*(k,x,y)=\arg\min_\ell\{\mu(k,x,y,\ell):x\leq\ell,y\leq k-\ell\}$ is the parameter of those rectangles offering
the best possible ratio $\mu^*=\mu(k,x,y,\ell^*)$ for $D^{x,y}_{[k]}$.
Thus, $\mu^*\cdot||D^{x,y}_{[k]}||=\frac{\binom{\ell^*}{x}+\binom{k-\ell^*}{y}}{\binom{\ell^*}{x}\binom{k-\ell^*}{y}}\bigl(\binom{k}{x}\binom{k-x}{y}\bigr)$
is a lower bound for the cost of the optimal solution.

Observe that this bound is indeed attainable by the greedy strategy,
since each set $(X,Y)$ with $|X|=x$ and $|Y|=y$, $X\cap Y=\emptyset$ is covered exactly by 
$\binom{k-x-y}{\ell^*-x}$ such rectangles (i.e. respects this number of such bipartitions),
thus considering the fractional covering $\eta(\ell^*)$ which uses
\emph{all} such bipartitions with multiplicity $1 / \binom{k - (x + y)}{\ell^* - x}$ we get a covering of $D^{x,y}_{[k]}$,
with total cost $\frac{1}{\binom{k - (x + y)}{\ell^* - x}}\bigl(\binom{\ell^*}{x}+\binom{k-\ell^*}{y}\bigr)\binom{k}{\ell^*}$
(that is, multiplicity$\times$weight of a rectangle$\times$number of these rectangles).
The last expression is the same as
$\frac{\binom{\ell^*}{x}+\binom{k-\ell^*}{y}}{\binom{\ell^*}{x}\binom{k-\ell^*}{y}}\bigl(\binom{k}{x}\binom{k-x}{y}\bigr)$,
since $\binom{k}{x}\binom{k-x}{y}\binom{k - (x + y)}{\ell^* - x}=\binom{k}{\ell^*}\binom{\ell^*}{x}\binom{k-\ell^*}{y}$:
both of these products calculate the number of possibilities to choose an $\ell^*$-element subset $L$ of a $k$-element set $K$,
and an $x$-element subset $X$ of $L$ as well as an $y$-element subset of $K-L$. The first formula achieves this by
choosing $X$ from $K$ first, then $Y$ from $K-X$, finally $L-X$ from $K-X-Y$, the second one by
choosing $L$ from $K$ first, then $X$ from $L$ and finally $Y$ from $K-L$. Thus, choosing all these bipartitions
with this multiplicity provides an optimal solution.

Note that for any fixed $\ell$, the weighted set covering problem
using only the bipartitions $(S,\overline{S})$ with $|S|=\ell$ is a uniform-cost,
i.e., an unweighted set covering problem.
On such a problem the greedy heuristic achieves a value
within a factor of $1 + \log \binom{k}{x} \binom{k}{y} \le 1 + 2k = \polylog(n)$ of
the optimum in the linear relaxation.
Therefore, it suffices to pick some $\ell$ and construct a greedy covering
using bipartitions into sets of size $\ell$ and $k - \ell$.
Our choice of $\ell$ in Lemma~\ref{l:drop-polylog} is $\ell = x + (k - x - y) / 2$,
and the argument above shows that the optimal choice, $\ell = \ell^*(k; x, y)$
will deliver an upper bound on $\ortwo(D_n)$ that is tight up to a polylogarithmic
factor, thus reducing the problem to a parametric optimization task.


\section{Application: size of regular expressions}

A {\it regular expression} over $\Sigma$ is a well-formed expression $r$
consisting of the symbols
$$ \epsilon, \emptyset, {\mbox{\tt (}}, {\mbox{\tt )}}, \mbox{\tt +}, \mbox{\tt *}, 
\text{ and } a \in \Sigma ,$$
with the usual semantics (e.g., as in \cite{Hopcroft}).

The {\it size} of a regular expression $r$ can be specified in a number
of different ways, but for our purposes, the easiest is the so-called
{\it alphabetic length}, which is the number of symbols in $r$ belonging
to $\Sigma$ \cite{Lee}.  For example, the alphabetic length of
\begin{equation}
r = a_0 a_1 + a_2 a_3 + (a_0+a_1)(a_2+a_3)
\label{ree}
\end{equation}
is $8$.

Given a regular language $L$ specified in some way (for example,
as the language accepted by a finite automaton),
it is, in general, quite difficult to
determine the size of the shortest regular expression specifying
$L$.  In fact, this problem is PSPACE-hard \cite{Meyer,Jiang} and
not even approximable within a factor of $o(n)$ \cite{Gramlich}
(unless P = PSPACE).

\subsubsection*{Extended example}

In this subsection we examine a specific family of finite languages, namely
$$L_n = \sum_{0 \leq i < j < n} a_i a_j,$$
over the alphabet $\Sigma_n = \{ a_0, a_1, \ldots, a_{n-1} \}$ of
size $n$, 
and we provide matching upper and lower bounds on
for the size of the shortest regular expression for it.
For example, for $n = 4$ this is the language
$$ L_4 = \{ a_0 a_1, a_0 a_2, a_0 a_3, a_1 a_2, a_1 a_3, a_2 a_3 \}.$$
Evidently one can produce a regular expression for $L_n$ of
length $n(n-1)$ by listing the elements of $L_n$, but it is possible
to do much better.  For example, the regular expression given in
\eqref{ree} specifies $L_4$ with alphabetic length $8$, as opposed to
length $12$ using the brute-force approach.

Our upper and lower bounds follow Corollary~\ref{cor:ft} in the main text.
For the lower bound, we relate the alphabetic length of regular expressions
to the cost of coverings of Boolean matrices;
for the upper bound, we provide a direct proof
to make the connection between regular expressions and coverings
more transparent.


We first show how to construct a small regular expression
for $L_n$ through a simple divide-and-conquer strategy.  We generalize $L_n$
to $L_{A,B} = \bigcup_{A \leq i < j \leq B} a_i a_j$ so that
$L_n = L_{0, n-1}$.  Then our divide-and-conquer solution is given
by
$$ L_{A,B} = L_{A,C} \ \cup\ L_{C+1,B} \ \cup\  
	\lbrace a_A + a_{A+1} + \cdots + a_C \rbrace \cdot
	\lbrace a_{C+1} + \cdots + a_B \rbrace ,$$
where $C = \lfloor (A+B)/2 \rfloor$.
The alphabetic length $t(n)$ of the regular expression so constructed satisfies the
recurrence
$t(1) = 0$ and $t(2n) = 2t(n) + 2n$ and
$t(2n+1) = t(n+1) + t(n) + 2n+1$.  
Now an easy induction proves that in fact $t(n) = s(n)$,
with
$s(n) = n(\lfloor \log_2 n \rfloor + 2) - 2^{\lfloor \log_2 n \rfloor + 1}$.

We now turn to the lower bound.
Let $r_n$ be a regular expression of shortest length for $L_n$ for $n \geq 2$.
Clearly we can assume that $r_n$ contains no occurrence of the
empty set symbol $\emptyset$.
Since $L_n$ is finite, we can also assume $r_n$ contains no occurrence
of $\mbox{\tt *}$.  So all the operators in $r_n$ are either union
or concatenation.  Consider any instance of concatenation, say 
$L_1 L_2$.  Then if either $L_1$ or $L_2$ contains strings of two
different lengths, the resulting concatenation would also, which is
impossible since $L_n$ contains only strings of length $2$.  So all
strings on one side of any concatenation are of the same length.
On the other hand, no strings can be of length $3$ or more, and if
one side contains only strings of length $0$ (the empty string) we could
simply omit the concatenation.  So in fact we may assume, without
loss of generality that any concatenation in $r_n$ looks like
$R\cdot C$, where both languages consist of subsets of $\Sigma_n$.  
Finally, every letter in $C$ must be numbered higher than all those
of $R$, for otherwise we would obtain a word not in $L_n$.
This means that we can write $r_n$ as
\begin{equation}
R_1 \cdot C_1 + R_2 \cdot C_2 + \cdots + R_t \cdot C_t 
\label{re}
\end{equation}
where we have inserted dots to make the concatenation explicit.
The alphabetic length of this expression is 
$\sum_{1 \leq i \leq t} (|R_i|+ |C_i|).$

We now create an integer program to minimize this length.
Define $\isetn = \lbrace 0,1, \ldots, n-1 \rbrace$ and
let $x_{R, C}$ for nonempty sets $R, C \subseteq \isetn$
be an indicator variable for the presence of the
term $R\cdot C$ in the expression \eqref{re}:  $1$ if it is present
and $0$ otherwise.  Our integer program is

\bigskip

\noindent\fbox{%
\begin{minipage}{4.5in}
\noindent minimize $\sum_{ {{R,C \text{ nonempty} } \atop {R, C \subseteq \isetn}}
	\atop{\max R < \min C}} (|R|+|C|) x_{R, C}$ \\
subject to the constraints \\
$x_{R,C} \in \{ 0,1 \}$ for nonempty $R, C \subseteq \isetn$ and
	$\max R < \min C$  \\
$\sum_{{i \in R}\atop{j \in C}} x_{R,C} \geq 1$ for nonempty  $R, C \subseteq \isetn$ and
        $\max R < \min C$ . 
\end{minipage}
}

\bigskip

The last constraint means that every string $a_i a_j$ with $i < j$ is covered
by at least one concatenation of sets.   Note that we write ``$\geq 1$'' in the
last group of inequalities instead of ``$=1$'', because we are not insisting
that our regular expression be unambiguous.

For example, if $n = 3$ then the integer program is

\bigskip

\noindent\fbox{%
\begin{minipage}{3.5in}
\noindent minimize $2 x_{0,1} + 2 x_{0,2} + 2 x_{1,2} + 
	3 x_{01,2} + 3 x_{0,12}$ \\
subject to the constraints \\
$x_{0,1}, x_{0,2}, x_{1,2}, x_{01,2}, x_{01,2} \in \{ 0,1 \} $ \\
$ x_{0,1} + x_{0,12} \geq 1$ \\
$ x_{0,2} + x_{01,2} + x_{0,12} \geq 1$ \\
$ x_{1,2} + x_{01, 2} \geq 1$.
\end{minipage}}
\bigskip

It is not difficult to see that our integer program, in fact,
is the weighted set covering formulation, from Section~\ref{s:techniques},
where the optimal value is $\ortwo(T_n)$ with $T_n$ the $n \times n$
full triangular matrix, as in Section~\ref{s:lt}.
So we can conclude from Corollary~\ref{cor:ft} that
the smallest alphabetic length of a regular expression for
the language $L_n$ is
$s(n) = n(\lfloor \log_2 n \rfloor + 2) - 2^{\lfloor \log_2 n \rfloor + 1}$.

\bigskip

In what follows, we illustrate the approach taken in the main text
by formulating the linear relaxation of the integer program above
and taking its dual. This follows Figure~\ref{f:prog} in Section~\ref{s:techniques}.

The integer program above is an instantiation of
the one in Figure~\ref{f:prog:integer}.
We now relax the constraints on the $x_{R,C}$ to be
$0 \leq x_{R,C} \leq 1$.  The dual linear program then has
variables $y_{i,j}$ corresponding to the string $a_i a_j$,
for $0 \leq i < j < n$; compare to Figure~\ref{f:prog:linear}.
The corresponding dual, as in Figure~\ref{f:prog:dual}, is 

\bigskip

\noindent\fbox{%
\begin{minipage}{5in}
\noindent maximize $\sum_{0 \leq i < j < n} y_{i,j}$ \\
subject to the constraints \\
$y_{i,j} \geq 0$ for $0 \leq i < j < n$ \\
$ \sum_{{i \in R}\atop{j \in C}} y_{i,j} \leq |R|+|C|$ for
nonempty $R, C \subseteq \isetn$ and $\max R < \min C$.  \\
\end{minipage}}

\bigskip

For example, for $n = 3$ the corresponding dual is

\bigskip

\noindent\fbox{%
\begin{minipage}{2.0in}
\noindent maximize $y_{0,1}+y_{0,2}+y_{1,2}$  \\
subject to the constraints \\
$y_{0,1} \geq 0$ \\
$y_{0,2} \geq 0$ \\
$y_{1,2} \geq 0$ \\
$y_{0,1} \leq 2$ \\
$y_{0,2} \leq 2$ \\
$y_{1,2} \leq 2$ \\
$y_{0,1}+y_{0,2} \leq 3$ \\
$y_{0,2}+y_{1,2} \leq 3$.
\end{minipage}}

\subsubsection*{General connection}

Whenever $L\subseteq\Sigma\Delta$ for the alphabets $\Sigma=\setm$ and $\Delta=\setn$,
and $M_L$ is its characteristic $m\times n$ matrix $M_{i,j}=1$ iff $ij\in L$,
then the following statements hold:
\begin{enumerate}
\item The value $\OR_2(M_L)$ coincides with the smallest possible alphabetic length
of a regular expression for $L$.
\item The value $\OR_2(M_L)$ also coincides with the size
of the smallest $\varepsilon$-free nondeterministic finite automaton (NFA) recognizing $L$.
\item The value $\OR(M_L)+m+n$ is an upper bound on the size
of the smallest
nondeterministic finite automaton with possible $\eps$-transitions ($\eps$-NFA)
recognizing $L$.
\end{enumerate}

The proof of the first statement follows the example above,
and the last two statements can be found in~\cite{ivanlelkesnagygyorgyturan14}.

\end{document}